\pdfoutput=1

\documentclass{llncs}

\usepackage[pdftex]{graphicx,color}
\usepackage{subfigure}
\usepackage[noadjust]{cite}
\usepackage{paralist}
\usepackage[shortlabels]{enumitem}

\usepackage{amsmath,amssymb}

\usepackage[hyperref]{ntheorem}

\usepackage{comment}
\usepackage[left=1.8in,right=1.8in,top=1.6in,bottom=1.6in]{geometry}

\theoremstyle{plain}
\theoremseparator{.}
\newtheorem{observation}[theorem]{Observation}
\newtheorem{lemm}[theorem]{Lemma}

\newcommand{\cK}[0]{\ensuremath{\mathcal{K}}}
\newcommand{\cFPT}[0]{\ensuremath{\mathsf{FPT}}}
\newcommand{\bigO}[1]{\ensuremath{\mathcal{O}(#1)}}

\newcommand{\cNP}[0]{\ensuremath{\mathsf{NP}}}

\newcommand{\cL}[0]{\ensuremath{\mathcal{L}}}

\newcommand{\bbR}[0]{\ensuremath{\mathbb{R}}}

\newcommand{\eps}[0]{\ensuremath{\varepsilon}}	
\newcommand{\without}[0]{\ensuremath{\backslash}}	
\newcommand{\wo}[0]{\ensuremath{\backslash}}	
\newcommand{\tp}[0]{\ensuremath{\top}}			

\newcommand{\midresize}[0]{\ensuremath{\, \middle| \,}}

\makeatletter
\def\model{\@ifnextchar[{\@modelWithTwo}{\@modelWithOne}}
\def\@modelWithTwo[#1]#2{
{(\text{#2}_{\text{#1}})}}
\def\@modelWithOne#1{{(\text{#1})}}
\makeatother

\makeatletter
\def\polytope{\@ifnextchar[{\@polytopeWithTwo}{\@polytopeWithOne}}
\def\@polytopeWithTwo[#1]#2{\mathcal{P}^{\text{#2}}_{\text{#1}}}
\def\@polytopeWithOne#1{{\mathcal{P}_{\text{#1}}}}
\makeatother

\newcommand{\projection}[2]{\text{Proj}_{#1}\left(#2\right)}
\newcommand{\projectionS}[2]{\text{Proj}_{#1}(#2)} 

\newtheorem{myToDo}{ToDo}

\usepackage{tikz}
\usepackage{pgfplots}
\pgfplotsset{major grid style={thick}}
\usetikzlibrary{arrows}
\usetikzlibrary{backgrounds}
\usetikzlibrary{shapes.geometric}
\usetikzlibrary{fit}
\usetikzlibrary{plotmarks}
\definecolor{green}{rgb}{0,0.39,0}

\tikzstyle{basicEdge} = [draw,thick,-]
\tikzstyle{dashedEdge} = [draw,thick,dashed,-]

\tikzstyle{arc} = [draw,thick,->] 
\tikzstyle{arcr} = [draw,thick,->] 
\tikzstyle{arcl} = [draw,thick,<-] 
\tikzstyle{arcrl} = [draw,thick,<->] 
\tikzstyle{forestEdgeTreewidth} = [draw,line width=2pt,-]
\tikzstyle{dottedEdgeTreewidth} = [draw,thick,densely dashed,-]
\tikzstyle{emptyNode}=[] 
\tikzstyle{circle}=[draw,thick,shape=circle,inner sep=0pt,minimum size=5.5mm]
\tikzstyle{circleBlack}=[draw,thick,shape=circle,inner sep=0pt,minimum size=4mm,fill=black]
\tikzstyle{rectangle}=[draw,thick,shape=rectangle,minimum size=4.5mm]
\tikzstyle{rectangleDotted}=[draw,dotted,thick,shape=rectangle,minimum size=4.5mm]

\tikzstyle{diamond}=[draw,thick,shape=rectangle,rotate=-45,minimum size=2.5mm] 

\usepackage{hyperref}

\begin{document}

\title{ILP formulations for the two-stage stochastic Steiner tree problem}
\author{
Bernd Zey
}

\institute{
Department of Computer Science, TU Dortmund, Germany\\
\email{bernd.zey@tu-dortmund.de}
}

\maketitle

\begin{abstract}
We give an overview of new and existing cut- and flow-based ILP formulations for the two-stage stochastic Steiner tree problem 
and compare the strength of the LP relaxations. 


\end{abstract}

\section{Introduction}

The {\em Steiner tree problem (STP)} is a classical network design problem:  
For an undirected graph $G=(V,E)$ with edge costs $c_e \in \bbR^{\geq 0}, \forall e\in E$, and a set of {\em terminals} $\emptyset \not=T\subseteq V$ it asks for a minimum cost edge set $E'\subseteq E$ such that $G[E']$ connects $T$. 
The decision problem of the STP is \cNP-complete \cite{Karp1972}, even in case of edge weights 1 and 2 \cite{BernPlassmann1989} 
or when the graph is planar \cite{GareyJohnson1977}.
It is solvable in polynomial time if the graph is series-parallel (partial 2-tree) \cite{WaldColbourn1983} and 
it is in \cFPT\ with the parameter being the treewidth $k$ (partial $k$-trees) 
\cite{ChimaniMutzelZeySTP-FPT}
or the number of terminals \cite{DreyfusWagner1972}. 
Moreover, the STP is approximable with a constant factor and the currently best ratio is $\ln(4) + \eps = 1.39$ \cite{ByrkaEtAlJournal}.
Moreover, ILP formulations and their polytopes have been studied intensely in the 1990's, see, e.g., 
\cite{ChopraRao1994,ChopraRao1994b,
ChopraTsai2001,
GoemansMyung1993,
KochMartinSTP,
PolzinDaneshmand2001}.

The two-stage stochastic Steiner tree problem  is a natural extension of the STP to a two-stage stochastic combinatorial optimization problem; for an introduction to stochastic programming see, e.g., \cite{BirgeLouveaux1997,ShapiroEtAlBookStochasticProgramming,KallWallaceBook}. 
In the first stage, today, it is possible to buy some ``profitable'' edges while the terminal set and the edge costs are subject to uncertainty.  
However, all possible outcomes are known and given by a set of scenarios. 
In the second stage, in the future, one of the given scenarios is realized and  additional edges have to be installed in order to 
connect the now known set of terminals. 
The objective is to make a decision about edges to be purchased in the first stage and in each scenario such that the terminal sets in each scenario are connected and the expected cost of the overall solution is minimized. 

Formally, the stochastic Steiner tree problem (SSTP) is defined as follows: 
We are given an  undirected graph $G=(V,E)$, first stage edge costs $c_e^0 \in\bbR^{\ge0}, \forall e\in E$, and a set of $K \ge 1$ scenarios with $\cK:=\{1,\ldots, K\}$.
Each scenario $k\in\cK$ is defined by its probability $p^k \in (0; 1]$, 
second-stage edge costs $c_e^k \in\bbR^{\ge0}, \forall e\in E$, 
and a set of terminals $\emptyset\not=T^k\subseteq V$. 
Thereby, it holds $\sum_{k\in\cK}p^k = 1$. 
A feasible solution consists of $K+1$ edge sets $E^0, \dots, E^K \subseteq E$ such that $G[E^0 \cup E^k]$ connects $T^k, \forall k\in\cK$. 
The objective is to minimize the expected cost  $\sum_{e\in E^0} c_e^0 + \sum_{k\in\cK} p^k \sum_{e\in E^k} c_e^k$. 

The \emph{expected cost of an edge} $e\in E$ is defined as $c_e^* := \sum_{k\in \cK} p^k c_e^k$.
W.l.o.g.\ one can assume that $c_e^0 < c_e^*, \forall e\in E$; 
otherwise, this edge would never be purchased in the first stage since it can be installed in every scenario at the same or cheaper cost.
On the other hand it is also valid to assume $c_e^0 > \min_{k\in\cK}\{p^k c_e^k\}, \forall e\in E$, since this edge would never be installed in any scenario.

Notice that for the SSTP the optimum first stage solution $E^0$ does not have to be connected.
In particular, it is easy to construct instances with the optimum first stage solution being a forest, cf.\ Figure \ref{figure:rsstp:wlog:assumption} and Figure \ref{figure:sstp:examples:disconnected:first:stage:directed} in Section \ref{section:sstp:ip:formulations:directed}.
However, fragmented solutions might be unreasonable in practical settings. 
For example, if new cables or pipes are installed in a city one would prefer starting at one point and connecting adjacent streets first and not by digging in several parts of the city simultaneously.

This leads to the \emph{rooted stochastic Steiner tree problem} (rSSTP) which is defined similarly to the SSTP. 
It additionally has a root node $r\in V$ which is a terminal in each scenario, i.e., $r\in T^k, \forall k \in \cK$. 
Then, a feasible solution again consists of $K+1$ edge sets $E^0, \dots, E^K \subseteq E$ such that $G[E^0 \cup E^k]$ connects $T^k, \forall k\in\cK$, but it is required that $G[E^0]$ is a tree containing $r$. 
As for the SSTP the objective is to minimize the expected cost.

Notice that the assumption $c_e^0 < c_e^*, \forall e\in E$, as for the SSTP, is not valid for the rSSTP due to the necessary first stage tree. 
This is shown by Figure \ref{figure:rsstp:wlog:assumption}; here, edge $e_2$ would be disabled in the first stage which prohibits the optimum solution.
By swapping first- and second-stage edge costs this example shows that this holds for 
assumption $c_e^0 > \min_{k\in\cK}\{p^k c_e^k\}$ as well. 

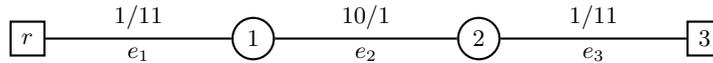
\begin{figure}[tb]
\centering
\begin{tikzpicture}[scale=1.0]

\node[rectangle] (r) at (1,10) {$r$};
\node[circle] (1) at (4,10) {$1$};
\node[circle] (2) at (7,10) {$2$};
\node[rectangle] (3) at (10,10) {$3$};

\draw[basicEdge] (r) -- node[above] {1/11} node[below] {$e_1$} (1);
\draw[basicEdge] (1) -- node[above] {10/1} node[below] {$e_2$} (2);
\draw[basicEdge] (2) -- node[above] {1/11} node[below] {$e_3$} (3);

\end{tikzpicture}

\caption{
A simple example for the SSTP where the optimum first stage solution is disconnected. 
There exists only one scenario (connect terminals $r$ and $3$) and edge costs for the first stage and the scenario, respectively, are written above the edges. 
The optimum solution selects edges $\{r,1\}, \{2,3\}$ in the first stage and $\{1,2\}$ in the scenario with overall cost 3. 
Interpreted as rSSTP instance this example shows that applying the assumption ``$c_e^0 < c_e^*, \forall e\in E$'' is not feasible, cf. text.  
For the rSSTP the optimum solution uses all edges in the first stage with overall cost 12. Disabling $e_2$ in the first stage would imply cost 13.
}
\label{figure:rsstp:wlog:assumption}
\end{figure}

\paragraph{Organization.}
We start in Section \ref{section:sstp:complexity:related:work} with an overview of the related work. 
Section \ref{section:ILP:formulations} introduces known and new ILP formulations based on undirected cuts and flows (Section \ref{section:undirected:formulations}), 
stronger semi-directed formulations by using orientations properties (Section \ref{section:sstp:ip:formulations:semi:directed}), 
and directed formulations for the rooted SSTP  (Section \ref{section:sstp:ip:formulations:directed}). 
In the last part, in Section \ref{section:sstp:ip:formulations:strength}, all described ILP formulations are compared by considering the strength of their LP relaxations.

\section{Related work}
\label{section:sstp:complexity:related:work}

\paragraph{Approximations.}
Although the STP allows constant factor approximations the stochastic problems are harder to approximate.
\cite{RaviSinha-ua-StochasticShortestPath} showed that the group Steiner tree problem, which is $\Omega(\log^{2-\eps}n)$-hard to approximate, can be reduced to the stochastic shortest path problem 
(a special case of the (r)SSTP). 
Nevertheless, in literature stochastic versions of the STP have been mostly investigated for approximation algorithms.
Due to the inapproximability results restricted versions have been considered to obtain approximation algorithms, 
e.g., by introducing a fixed and/or uniform inflation factor or a global terminal (a vertex being a terminal in all scenarios). 
Moreover, different models of scenario representations are used.  
Here, we concentrate on the finite/polynomial scenario model where the random variables of the stochastic problems are assumed to have finite support. 
Other publications consider the black box/oracle model.  
For an overview of these concepts see, e.g., \cite{ShmoysSwamy2006Approx2StageStochOpt}. 

\cite{GuptaEtAl2007} consider the SSTP with $K$ inflation factors and a global terminal
and  present a 40-approximation. 
\cite{GuptaKumar2009StochasticSteinerForest} consider the problem with a uniform fixed inflation factor but without global terminal and describe a constant factor approximation. 

For the black box/oracle model there exist several approximation algorithms 
which are based on the idea of scenario sampling. 
\cite{ImmorlicaEtAlApproxStochCombOpt} present an $\bigO{\log n}$-\-app\-ro\-xi\-mation algorithm for 
a problem which is restricted by a uniform inflation factor. 
\cite{GuptaEtAl2004BoostedSampling,GuptaEtAl2012SamplingApproximationCostSharing} 
introduce the concept of boosted sampling 
and 
consider the problem with a global terminal and a uniform inflation factor; their approximation algorithm has a ratio of $3.55$.
A similar problem is considered by \cite{ShmoysSwamy2006Approx2StageStochOpt} who present a 4-approximation.
\cite{GuptaPal2005StochasticSteinerTreeWithoutRoot} approximate a problem without global terminal.
This problem has a fixed uniform inflation factor and the presented algorithm has a ratio of $12.6$.  

\cite{GuptaEtAl-SSTP-NonUniformInflation} consider no uniform inflation factor  
but there are only two cost functions for the first stage edges and one for the second-stage edges.
The problem is shown to be at least $\Omega(\log\log n)$-hard 
and an approximation algorithm with a polylogarithmic approximation ratio is given.

\paragraph{Related publications.}
Among others, the approach by \cite{GuptaEtAl2007} is based on a primal-dual scheme where an undirected cut- and flow-based formulation is used. 
\cite{BomzeEtAl2010} describe a stronger semi-directed cut-based formulation for the SSTP, apply a Benders decomposition/two-stage branch\&cut approach, and present an experimental study. 
\cite{HeuristicSSTP-Dimacs2014} describe a heuristic for the SSTP which is compared to the exact approach experimentally. 
\cite{LjubicMutzelZeySSNDPpaper,LjubicMutzelZeySSNDParticle} expand the SSTP to stochastic survivable network design problems and undirected and semi-directed cut-based  formulations are introduced. 

Last but not least, fixed parameter tractable algorithms are described for the stochastic problems with parameter overall number of terminals \cite{KurzMutzelZey2013} and on partial 2-trees with parameter number of scenarios \cite{BoeklerZeyMutzel2012}.

\section{ILP formulations}
\label{section:ILP:formulations}

We start by introducing undirected cut- and flow-based formulations for the SSTP in Section \ref{section:undirected:formulations}. 
Afterwards we consider semi-directed models in Section \ref{section:sstp:ip:formulations:semi:directed}. 
The rooted version can be modeled by stronger directed formulations which are described in Section \ref{section:sstp:ip:formulations:directed}.
Finally, Section \ref{section:additional:constraints} deals with additional constraints for the described models. 

\paragraph{Notations and definitions.} 
We always use the upper index 0 to indicate the first stage and indices $1, \ldots, K$ for the $K$ scenarios,
e.g., $x^0$ is the vector of undirected edge variables of the first stage and $y^1$ and $z^k$ are directed arc variables of the first and $k$th scenario, respectively. 
To shorten the notation we use the superscript $1\dots K$ to abbreviate $K$ combined scenario vectors:  
the vector $x^{1\dots K}$ is the transposed concatenation of the vectors $x^1, \dots, x^K$, i.e., 
$x^{1\dots K} = ((x^1)^\tp, \dots, (x^K)^\tp)^\tp$.
We use $0\dots K$ analogously.
Moreover, if, e.g., $x^0$ and $y^{1\dots K}$ are variable vectors 
we abbreviate the vector  
$((x^0)^\tp, (y^{1\dots K})^\tp)^\tp$ by 
$(x^0, y^{1\dots K})$.

For an undirected weighted graph $G=(V,E)$ with edge cost $c_e, \forall e \in E$,  the {\em bidirection} of $G$ is the directed graph $\bar G = (V, A)$ with the arc set $A:=\bigcup_{\{i,j\}\in E} \{(i,j), (j,i)\}$ and
arc costs $c_{ij} = c_{ji} = c_e, \forall e=\{i,j\}\in E$. 
We use the common abbreviations for undirected and directed cuts for a vertex set $\emptyset\not= S \subset V$:
$\delta(S) = \{e\in E \mid |e\cap S| = 1\}$ and $\delta^-(S) = \{(i,j)\in A\mid i\not\in S, j\in S\}$. 
Moreover, if $x$ is a variable vector for undirected edges and $z$ for directed arcs 
we use $x(E') = \sum_{e\in E'} x_e$ and $z(A') = \sum_{a\in A'} z_a$.

In the semi-directed formulations each scenario $k\in\cK$  has a designated root vertex $r^k\in T^k$. 
Then, let $T^k_r := T^k\wo\{r^k\}$ and $V_r^k := V\setminus \{r^k\}$.
Moreover, let $t_r^* := \sum_{k\in\cK}|T_r^k|$.
In the directed formulations with root node $r$ we have $V_r := V\wo\{r\}$ and $T_r^k := T^k\wo\{r\}, \forall k\in\cK$.

\subsection{Undirected formulations}
\label{section:undirected:formulations}

\paragraph{Undirected cut formulation.}
The following IP is a formulation based on undirected cuts and was frequently considered in literature, e.g., by  \cite{GuptaEtAl2007}.
It is the classical expansion of the undirected cut formulation for the STP, see, e.g., \cite{KochMartinSTP,PolzinDaneshmand2001}. 
Binary decision variables for the first stage edges are denoted by $x_e^0, \forall e\in E$, and scenario edges of the $k$th scenario by $x_e^k, \forall e\in E, \forall k\in \cK$. 
The objective is to minimize the expected cost which is the sum of the selected first stage edges plus the sum of second-stage edges weighted by the scenario probability.
\begin{align}
\model[uc]{SSTP}\, 
\min\, &\sum_{e\in E}c_e^0 x_e^0 + \sum_{k\in\cK} p^k \sum_{e\in E}c_e^k x_e^k \nonumber \\
\text{s.t. }  
  (x^0 + x^k)(\delta(S)) &\ge 1 \hspace{25pt}\forall k\in\cK, \forall S\subseteq V\colon \emptyset\not=T^k\cap S\not=T^k \label{SSTP:ucut:undirected:cuts}\\
  x^0 &\in \{0, 1\}^{|E|} &&\\
  x^{1 \dots K} &\in \{0, 1\}^{|E|\cdot K} &&
\end{align}

Constraints \eqref{SSTP:ucut:undirected:cuts} are undirected cuts ensuring the connectivity of each scenario terminal set.
Thereby, first-stage and second-stage edges can be used to satisfy a cut $S\subseteq V$; 
we use the notation  $(x^0 + x^k)(\delta(S)) = \sum_{e\in\delta(S)} x_e^0 + x_e^k$.

\paragraph{Undirected flow formulation.}
Here, we present a similar model to the one introduced by \cite{GuptaEtAl2007}.
We modify the model such that we have a flow only in the second stage. 
Thereby, the flow can be constructed by using selected first-stage or second-stage edges. 

We again use  variables $x^0$ and $x^k, \forall k\in\cK$, for modeling the solution edges.
Moreover, the bidirection with arc set $A$ is considered and a flow $f$ is computed in each scenario $k\in\cK$ from a designated root node $r^k\in T^k$ to each terminal.
We use variables $f_{ij}^{k,t}$ for each scenario $k\in\cK$, arc $(i,j)\in A$, and terminal $t\in T^k_r$. 
The undirected flow model for the SSTP then reads as follows:
\begin{align}
\model[uf]{SSTP}\, 
\min\, &\sum_{e\in E}c_e^0 x_e^0 + \sum_{k\in\cK} p^k \sum_{e\in E}c_e^k (x_e^k - x_e^0) \nonumber\\
\text{s.t. }  
%
x_e^0 + x_e^k &\ge f_{ij}^{k,t}, \nonumber\\
x_e^0 + x_e^k &\ge f_{ji}^{k,t}  
	\hspace{20pt} \forall k \in\cK, \forall e=\{i,j\}\in E, \forall t\in T^k_r \label{sstp:uflow:capacity:constraint:scenario}\\
\sum_{(h,i) \in A} f_{hi}^{k,t} - \sum_{(i, j) \in A} f_{ij}^{k,t} &=
\left.\begin{cases}
-1, & \text{if } i = r^k\\
1, & \text{if } i = t\\
0, & \text{otherwise}
\end{cases} \right\}
\begin{array}{l}
\forall k\in\cK, \forall t\in T^k_r,\\ \forall i \in V 
\end{array}
\label{sstp:uflow:flow:conservation} \\
f &\in [0,1]^{|A|\cdot t_r^*} \label{sstp:uflow:nonnegativity:f}\\
x^0 &\in \{0, 1\}^{|E|} \label{sstp:uflow:integrality:x0}\\
x^{1 \dots K}  &\in \{0, 1\}^{|E|\cdot K} \label{sstp:uflow:integrality:xk}
\end{align}

In this model there has to be one unit of flow in each scenario from the root to each terminal.
This is enforced by the flow conservation constraints \eqref{sstp:uflow:flow:conservation};
the root has one outgoing flow (first case), the terminal one ingoing flow (second case), and for all other vertices the ingoing flow equals the outgoing flow.
Edges which are used for routing the flow are selected as solution edges by the capacity constraints \eqref{sstp:uflow:capacity:constraint:scenario}, either as first-stage or as second-stage edges.
It is easy to see that the formulation $\model[uf]{SSTP}$ is valid and that it is equivalent to the one introduced by \cite{GuptaEtAl2007}.

Due to  the deterministic STP it is not surprising that the cut-based formulation is equivalent to the flow formulation, cf.\ Section \ref{section:sstp:ip:formulations:strength}. 
However, there exist stronger formulations by using orientation properties.

\subsection{Semi-directed formulations}
\label{section:sstp:ip:formulations:semi:directed}

\paragraph{Semi-directed cut formulations.}

In the following we introduce three semi-directed cut-based formulations for the SSTP. 
All models are based on the application of orientation properties like in the directed cut formulation for the  STP.
However, edge variables $x^0$ for the first stage remain undirected in all semi-directed formulations.
As will be discussed at the beginning of Section \ref{section:sstp:ip:formulations:directed}, using a directed first stage is difficult and no stronger formulation is known.
On the other hand, it is possible to consider the bidirected input graph $\bar G=(V,A)$ in the second stage.

In the first semi-directed model we use arc variables $z^k_a, \forall a\in A, \forall k\in\cK$. 
We search for a first-stage edge set $E^0$ and second-stage arc sets $A^1, \dots, A^K$ such that $E^0 \cup A^k$ contains a semi-directed path from a designated terminal $r^k\in T^k$ to each terminal in $T_r^k$, for all scenarios $k\in\cK$.
In other words, $A^0 \cup A^k$ has to contain a feasible arborescence for all scenarios $k\in\cK$, with $A^0 := \bigcup_{\{i,j\}\in E^0} \{(i,j), (j,i)\}$.

To shorten the notation we write $(x^0+z^k)(\delta^-(S)) := x^0(\delta(S)) + z^k(\delta^-(S)) = \sum_{(i,j)\in\delta^-(S)} x^0_{\{i,j\}} + z_{ij}^k$ for semi-directed cuts.
%
\begin{align}
\model[sdc1]{SSTP}\, 
\min\, &\sum_{e\in E} c_e^0 x_e^0 + \sum_{k\in\cK} p^k \sum_{e=\{i,j\}\in E} c_e^k (z_{ij}^k + z_{ji}^k) \nonumber \\
\text{s.t. }  
(x^0 + z^k)(\delta^-(S)) &\ge 1 \hspace{25pt}\forall k\in\cK, \forall S\subseteq V_r^k\colon S\cap T_r^k \not=\emptyset \label{SSTP:sdcut1:semi:directed:cuts}\\
x^0 &\in \{0, 1\}^{|E|} \\
z^{1\dots K} &\in \{0, 1\}^{|A|\cdot K}
\end{align}

This first formulation uses semi-directed cuts, i.e., each cut \eqref{SSTP:sdcut1:semi:directed:cuts} for scenario $k\in\cK$ can be fulfilled by first-stage edges or by second-stage arcs from this scenario.

\begin{lemm}
Formulation $\model[sdc1]{SSTP}$ models the stochastic Steiner tree problem correctly.
\end{lemm}
\begin{proof}
Let $\tilde E^0, \tilde E^1, \dots, \tilde E^K$ be an optimum solution for the stochastic Steiner tree problem.
Since this solution connects all terminals in all scenarios we can easily find 0/1-values for $x^0$ and $z^k, \forall k\in\cK$, respectively, by using exactly the edges $\tilde E^0, \dots, \tilde E^K$ such that there is a semi-directed path from $r^k$ to each terminal in $T_r^k, \forall k\in\cK$.

On the other hand, due to constraints \eqref{SSTP:sdcut1:semi:directed:cuts} an optimum solution $(\tilde x^0, \tilde z^{1\dots K})$ to $\model[sdc1]{SSTP}$ connects the designated root node $r^k$ with semi-directed paths to each terminal in $T_r^k$, for all scenarios $k\in\cK$. 
Hence, using the selected undirected first-stage edges plus the undirected counterparts of the second-stage arcs gives a feasible solution to the SSTP with the same objective value.
\qed
\end{proof}

In formulation  $\model[sdc1]{SSTP}$ a selected first-stage edge fulfills all related semi-directed cuts. 
Hence, in the extreme case when all terminals are connected via first-stage edges this model is not stronger than the undirected model.

This drawback is overcome by the second semi-directed formulation \cite{BomzeEtAl2010}.
It is based on additional capacity constraints which enforce that selected first-stage edges have to be incorporated into the second-stage solution:  
Each selected first-stage edge has to be oriented such that a feasible arborescence is established in each scenario.
Due to this change, the cut constraints are now purely directed and contain only second-stage arc variables $y^{1\dots K}$.
Because of the different meaning of the second-stage arc variables we use the identifier $y^{1\dots K}$ instead of $z^{1\dots K}$ as in $\model[sdc1]{SSTP}$.
The second semi-directed cut formulation for the SSTP reads as follows:
\begin{align}
\model[sdc2]{SSTP}\, 
\min\, &\sum_{e\in E} c_e^0 x_e^0 + \sum_{k\in\cK} p^k \sum_{e=\{i,j\}\in E} c_e^k (y_{ij}^k + y_{ji}^k - x_e^0) \nonumber\\
\text{s.t. }  
y^k(\delta^-(S)) &\ge 1 
	\hphantom{x_e^0}\hspace{25pt}\forall k\in\cK, \forall S\subseteq V_r^k\colon S\cap T_r^k \not=\emptyset \label{SSTP:sdcut2:directed:cuts}\\
y_{ij}^k + y_{ji}^k &\ge x_e^0 
	\hphantom{1}\hspace{25pt}\forall k\in\cK, \forall e=\{i,j\}\in E \label{SSTP:sdcut2:capacity:first:second:stage}\\
x^0 &\in \{0, 1\}^{|E|}  \label{SSTP:sdcut2:x0:integrality} \\
y^{1\dots K} &\in \{0, 1\}^{|A|\cdot K} \label{SSTP:sdcut2:yk:integrality}
\end{align}

This formulation is basically a union of $K$ directed Steiner tree formulations joined by the first stage through capacity constraints \eqref{SSTP:sdcut2:capacity:first:second:stage}.
Compared to the previous cut-based formulations the objective function contains a corrective term for subtracting the additional cost that results from these constraints.

\begin{lemm}[\cite{BomzeEtAl2010}]
Formulation $\model[sdc2]{SSTP}$ models the stochastic Steiner tree problem correctly.
\end{lemm}
\begin{proof}
An optimum solution $\tilde E^0, \tilde E^1, \dots, \tilde E^K$  to the SSTP can be easily translated into a feasible solution for model $\model[sdc2]{SSTP}$ by using the edge set $\tilde E^0 \cup \tilde E^k$ for finding a feasible arborescence in each scenario $k\in\cK$; 
then let variables $x^0$ represent $\tilde E^0$ and set arc variables $y^k$ according to the arborescences, $\forall k\in\cK$.

Contrarily, due to the correctness of the directed cut formulation for the deterministic STP
an optimum solution $(\tilde x^0, \tilde y^{1\dots K})$ to $\model[sdc2]{SSTP}$ contains an $r^k$-rooted arborescence in each scenario $k\in\cK$. 
Hence, $\tilde E^0, \tilde E^1, \dots, \tilde E^K$, with $\tilde E^0 := \{e\in E\mid \tilde x_e^0 = 1\}$ and $\forall k\in\cK\colon \tilde E^k := \{e=\{i,j\}\in E\mid \tilde y_{ij}^k = 1 \vee \tilde y_{ji}^k = 1\}$, is a feasible solution with the same objective value.
\qed
\end{proof}

Let $(\text{SSTP}_\text{sdc2}^{\text{rel}:x^0})$ denote formulation $\model[sdc2]{SSTP}$ with the integrality constraint \eqref{SSTP:sdcut2:x0:integrality} being relaxed to $x^0 \in [0,1]^{|E|}$.

\begin{lemm}[\cite{BomzeEtAl2010}]
\label{sstp:sdcut2:integrality}
The optimum solution to $(\mathit{SSTP}_\mathit{sdc2}^\mathit{rel:x^0})$ is integer.
\end{lemm}
\begin{proof}
Assume there exists an optimum solution $(\tilde x^0, \tilde y^{1\dots K})$ to $(\text{SSTP}_\text{sdc2}^{\text{rel}:x^0})$ that is non-integer.
Let variable $\tilde x_e^0$ corresponding to edge $e = \{i,j\} \in E$ be fractional, i.e., $0 < \tilde x_e^0 < 1$. 
The term in the objective function corresponding to edge $e$ is:
\begin{align}
& c_e^0 \tilde x_e^0 + \sum_{k\in\cK} p^k c_e^k (\tilde y_{ij}^k + \tilde y_{ji}^k - \tilde x_e^0) \nonumber\\
=\, &c_e^0 \tilde x_e^0 -   \sum_{k\in\cK} p^k c_e^k \tilde x_e^0 + \sum_{k\in\cK} p^k c_e^k (\tilde y_{ij}^k + \tilde y_{ji}^k) \nonumber \\
=\, & (c_e^0 - c_e^*) \tilde x_e^0 + \sum_{k\in\cK} p^k c_e^k (\tilde y_{ij}^k + \tilde y_{ji}^k)\nonumber
\end{align}

In case $c_e^0 < c_e^*$ set $\tilde x_e^0 := 1$ and 
if $c_e^0 > c_e^*$ set $\tilde x_e^0 := 0$. 
In both cases the resulting solution is still feasible: 
Constraint \eqref{SSTP:sdcut2:capacity:first:second:stage} together with the integrality of $y^{1\dots K}$ ensures that for all scenarios $k\in\cK$ it holds $\tilde y_{ij}^k + \tilde y_{ji}^k \ge 1$ and hence,  \eqref{SSTP:sdcut2:capacity:first:second:stage} is still satisfied. 
Moreover, the objective value improves which is a contradiction.

In case $c_e^0 = c_e^*$ variable $x_e^0$ has coefficient 0 in the objective function and can be fixed to $\tilde x_e^0 := 0$.  
\qed
\end{proof}

We like to shortly revisit formulation $\model[$\text{uc}$]{SSTP}$ based on undirected cuts. 
Notice that by adding similar capacity constraints $x_e^k \geq x_e^0, \forall k\in \cK, \forall e\in E$, 
the undirected cuts \eqref{SSTP:ucut:undirected:cuts} contain only second-stage variables, as in model $\model[$\text{sdc2}$]{SSTP}$. 
Moreover, it is possible to relax the first-stage variables to $x^0\in[0,1]^{|E|}$ without violating overall integrality; 
the proof is very similar to the one of Lemma \eqref{sstp:sdcut2:integrality}. 
On the other hand, these modifications do not influence the strength of the LP relaxation and this formulation is as strong as $\model[$\text{uc}$]{SSTP}$.

We close the discussion on semi-directed cut-based formulations by rewriting the objective function of $\model[sdc2]{SSTP}$. 
By moving the first-stage variables to the first sum gives the following formulation $\model[$\text{sdc2}^*$]{SSTP}$ (\cite{BomzeEtAl2010}):
\begin{align}
\model[$\text{sdc2}^*$]{SSTP} \min \sum_{e\in E} (c_e^0 - c_e^*) x_e^0 + &\sum_{k\in\cK} p^k \sum_{e=\{i,j\}\in E} c_e^k (y_{ij}^k + y_{ji}^k) \nonumber\\
\text{s.t. } (x^0,y^{1\ldots K}) \text{ satisfies }&\text{\eqref{SSTP:sdcut2:directed:cuts}--\eqref{SSTP:sdcut2:yk:integrality}} \nonumber
\end{align}

Obviously, $\model[$\text{sdc2}^*$]{SSTP}$ and $\model[sdc2]{SSTP}$ are identical. 
However, when the model gets decomposed with Benders' decomposition the modified objective function does matter, cf.\ \cite{BomzeEtAl2010}.
Then, the master problem of formulation $\model[$\text{sdc2}^*$]{SSTP}$ has negative coefficients (since $c_e^* > c_e^0$) whereas the coefficients in the master problem of $\model[sdc2]{SSTP}$ are non-negative. 
Moreover, this change affects the primal and dual subproblems and in particular, the generated 
optimality cuts.

\paragraph{Semi-directed flow formulation.}

The flow formulation can be strengthened as in the deterministic setting.
One simply has to enforce that a selected undirected edge cannot be used for routing flow in both directions at the same time, i.e., for one commodity.
Therefore, directed arc variables $y^k, \forall k\in \cK$, are used and constraints \eqref{sstp:uflow:capacity:constraint:scenario} are replaced by the stronger constraints \eqref{sstp:sdflow:capacity}.
To highlight the connection to formulation  $\model[sdc2]{SSTP}$ we use the same capacity constraints \eqref{sstp:sdflow:capacity:first:second:stage}.
%
\begin{align}
\model[sdf]{SSTP}\,
\min\, &\sum_{e\in E}c_e^0 x_e^0 + \sum_{k\in\cK} p^k \sum_{e=\{i,j\}\in E}c_e^k (y_{ij}^k + y_{ji}^k - x_e^0)\nonumber \\
\text{s.t. }
f \text{ satisfies } &\text{\eqref{sstp:uflow:flow:conservation}}\nonumber\\
y_{ij}^k &\ge f_{ij}^{k,t}  
	\hphantom{x_e^0}\hspace{25pt} \forall k \in\cK, \forall (i,j)\in A, \forall t\in T^k_r \label{sstp:sdflow:capacity} \\
y_{ij}^k + y_{ji}^k &\ge x_e^0 
	\hphantom{f_{ij}^{k,t}}\hspace{25pt} \forall k \in\cK, \forall e=\{i,j\}\in E \label{sstp:sdflow:capacity:first:second:stage} \\
f &\in [0,1]^{|A|\cdot t_r^*} \label{sstp:sdflow:nonnegativity:f}\\
 x^0 &\in \{0, 1\}^{|E|} \label{sstp:sdflow:integrality:x0}\\
  y^{1\dots K} &\in \{0, 1\}^{|A|\cdot K} && \label{sstp:sdflow:yk:integrality}
\end{align}

Formulation $\model[sdf]{SSTP}$ is the equivalent to  $\model[sdc2]{SSTP}$:
instead of satisfying directed cuts one has to find a feasible flow in each scenario and 
moreover, the scenarios are linked by the first-stage and  capacity constraints \eqref{sstp:sdflow:capacity:first:second:stage}.

\begin{observation}
Formulation $\model[sdf]{SSTP}$ models the stochastic Steiner tree problem correctly.
\end{observation}

%
\subsection{Directed formulations}
\label{section:sstp:ip:formulations:directed}

Formulating the SSTP with a directed first stage causes difficulties when first-stage solutions are disconnected.
Consider Figure \ref{figure:sstp:examples:disconnected:first:stage:directed} which depicts such an example.
Here, the optimum first-stage solution is disconnected as shown in Figure \ref{figure:sstp:examples:disconnected:first:stage:directed} (a). 
The optimum arborecences of the two scenarios are given in (b). 
In particular, edge $e_4$ is used in direction $(3,4)$ in the first and direction $(4,3)$ in the second scenario.
Hence, already fixing an orientation in the first stage omits an optimum scenario solution---or at least, makes the corresponding solution more expensive.

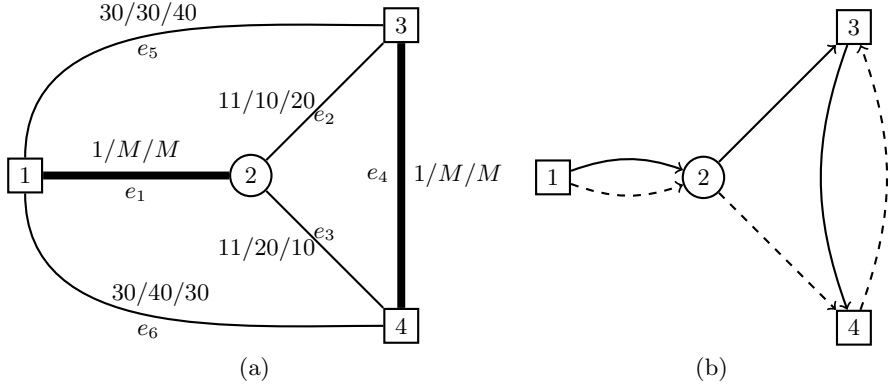
\begin{figure}[tb]
\centering
\subfigure[]{
\begin{tikzpicture}[scale=1.0]

\node[rectangle] (1) at (9,10) {$1$};
\node[circle] (2) at (12,10) {$2$};
\node[rectangle] (3) at (14,12) {$3$};
\node[rectangle] (4) at (14,8) {$4$};

\draw[basicEdge,line width=3pt] (1) -- node[above] {1/$M$/$M$} node[below] {$e_1$} (2);
\draw[basicEdge] (2) -- node[left] {11/10/20} node[below] {$e_2$} (3);
\draw[basicEdge] (2) -- node[left] {11/20/10} node[above] {$e_3$} (4);

\draw[basicEdge,line width=3pt] (3) -- node[right] {1/$M$/$M$} node[left] {$e_4$} (4);

\draw[basicEdge] (1) to [out=90,in=180] node[above] {30/30/40} node[below] {$e_5$} (3);
\draw[basicEdge] (1) to [out=270,in=180] node[above] {\quad 30/40/30} node[below] {$e_6$} (4);

\end{tikzpicture}
}
\subfigure[]{
\begin{tikzpicture}[scale=1.0]

\node[rectangle] (1) at (10,10) {$1$};
\node[circle] (2) at (12,10) {$2$};
\node[rectangle] (3) at (14,12) {$3$};
\node[rectangle] (4) at (14,8) {$4$};

\draw[arc] (1) to [out=20,in=160] node[above] {} (2);
\draw[arc] (2) -- node[above] {} (3);
\draw[arc] (3) to [out=250,in=110] node[above] {} (4);

\draw[arc,dashed] (1) to [out=340,in=200] node[above] {} (2);
\draw[arc,dashed] (2) -- node[above] {} (4);
\draw[arc,dashed] (4) to [out=70,in=290] node[above] {} (3);

\end{tikzpicture}
}
\caption{
(a) 
An SSTP instance with two equally probable scenarios with identical terminal set \{1,3,4\};
the edge costs for the first stage and the two scenarios are written next to the edges (i.e., first stage/first scenario/second scenario) with $M$ being a sufficiently large positive value.
The optimum solution edges of the first stage are highlighted by thick edges; scenario 1 and 2 additionally purchase edge $e_2$ and $e_3$, respectively.
The optimum solution has cost $1 + 1+ 0.5\cdot10 + 0.5\cdot10 = 12$.
(b) Minimum arborescences ($y^k$-values) in the scenarios for formulation $(\text{SSTP}_{\text{sdc2}})$. 
Solid arcs represent the first and dashed arcs the second scenario.}
\label{figure:sstp:examples:disconnected:first:stage:directed}
\end{figure}

\paragraph{Directed cut formulations for the rSSTP.}
While we are not aware of a fully directed and stronger cut-based formulation for the SSTP the rooted version of the SSTP permits a model with directed cuts only.
For the following formulations we again consider the weighted bidirection $\bar G=(V,A)$ of the input graph. 

The first formulation is called $\model[dc1]{rSSTP}$; afterwards, we introduce two more formulations $\model[dc2]{rSSTP}$ and $\model[$dc2^*$]{rSSTP}$, respectively, similar to the semi-directed case.
We use directed arc variables $z^0$ and $z^k$ for the first and second stage in scenario $k\in\cK$, respectively.

Constraints \eqref{rSSTP:dcut1:scenario:cuts} are directed cuts ensuring a feasible arborescence in each scenario consisting of first and second-stage arcs.
Moreover, the additional directed cuts \eqref{rSSTP:dcut1:cuts:first:stage} are used to enforce the required first-stage tree.
%
\begin{align}
\model[dc1]{rSSTP}\, 
\min\, &\sum_{a\in A} c_a^0 z_a^0 + \sum_{k\in\cK} p^k \sum_{a\in A} c_a^k z_a^k \nonumber\\
\text{s.t. }  
z^0(\delta^-(S)) &\ge z^0(\delta^-(v)) 
	\hspace{21pt}\forall \emptyset\not=S\subseteq V_r, \forall v\in S \label{rSSTP:dcut1:cuts:first:stage} \\
(z^0 + z^k)(\delta^-(S)) &\ge 1 
	\hphantom{z(\delta^-(v))}\hspace{20pt}\forall k\in\cK, \forall S\subseteq V_r\colon S\cap T_r^k \not=\emptyset \label{rSSTP:dcut1:scenario:cuts}\\
z^0 &\in \{0, 1\}^{|A|} \\
z^{1\dots K} &\in \{0, 1\}^{|A|\cdot K} \label{SSTP:dcut1:yk:integrality}
\end{align}

\begin{lemm}
Formulation $\model[dc1]{rSSTP}$ models the rooted stochastic Steiner tree problem correctly.
\end{lemm}
\begin{proof}
Let $\tilde E^0, \tilde E^1, \dots, \tilde E^K$ describe an optimum rSSTP solution.
Since $\tilde E^0$ induces a tree the edges can be oriented from the root $r$ outwards. 
Then, it is clear that for each scenario $k\in\cK$ the edge set $\tilde E^k$ can be oriented such that $\tilde E^0\cup \tilde E^k$ contains an arborescence with directed paths from $r$ to each terminal. 
This orienting procedure gives a solution to $\model[dc1]{rSSTP}$.

On the other hand, an optimum solution to $\model[dc1]{rSSTP}$ guarantees that every terminal is reachable by a directed path from the root node due to constraints \eqref{rSSTP:dcut1:scenario:cuts}.
Moreover, constraints \eqref{rSSTP:dcut1:cuts:first:stage} plus the objective function ensure that the first stage is a tree rooted at $r$.
Hence, the related undirected edges yield a feasible solution to the rSSTP.
\qed
\end{proof}

It is  possible to use the same idea leading to the semi-directed formulation $\model[sdc2]{SSTP}$ for another directed formulation for the rSSTP.
The variable identifier  for the first-stage arcs is $z^0$ and the arc variables for the $K$ scenarios are $y^{1\dots K}$. 
Again, we use identifier $y$ due to the different meaning: scenario arcs already contain selected first-stage arcs.
\begin{align}
\model[dc2]{rSSTP}\, 
\min\, &\sum_{a\in A} c_a^0 z_a^0 + \sum_{k\in\cK} p^k \sum_{a\in A} c_a^k (y_a^k - z_a^0) \nonumber\\
\text{s.t. }
z^0(\delta^-(S)) &\ge z^0(\delta^-(v)) 
	\hphantom{1z_{ij^0}}\forall \emptyset\not=S\subseteq V_r, \forall v\in S \label{rSSTP:dcut2:cuts:first:stage} \\
y^k(\delta^-(S)) &\ge 1 
	\hphantom{z_{ij^0}z^0(\delta^-(v))}\forall k\in\cK, \forall S\subseteq V_r\colon S\cap T_r^k \not=\emptyset \label{rSSTP:dcut2:scenario:cuts}\\
y_{ij}^k &\ge z_{ij}^0 
	\hphantom{z_0z^0(\delta^-(v))}\forall k\in\cK, \forall (i,j)\in A \label{rSSTP:dcut2:capacity:first:second:stage} \\
z^0 &\in \{0, 1\}^{|A|}  \label{rSSTP:dcut2:z0:integrality}\\
y^{1\dots K} &\in \{0, 1\}^{|A|\cdot K} \label{rSSTP:dcut2:yk:integrality}
\end{align}

Constraints \eqref{rSSTP:dcut2:cuts:first:stage} are identical to constraints \eqref{rSSTP:dcut1:cuts:first:stage}  in $\model[dc1]{rSSTP}$ and model the first-stage tree. 
Capacity constraints \eqref{rSSTP:dcut2:capacity:first:second:stage} enforce the selection of used first-stage arcs in each scenario.
Again, the objective function contains a corrective term for the additional cost.  
Then, the directed cuts \eqref{rSSTP:dcut2:scenario:cuts} in the scenarios contain only variables $y$.

\begin{observation}
Formulation $\model[dc2]{rSSTP}$ models the rooted stochastic Steiner tree problem correctly. 
\end{observation}

The objective function of model $\model[dc2]{rSSTP}$ can be rewritten analogously to the semi-directed formulation. 
We call the resulting formulation $\model[$\text{dc2}^*$]{rSSTP}$ which is equivalent to $\model[dc2]{rSSTP}$ but the change in the objective function matters when a decomposition is applied.
\begin{align}
\model[$\text{dc2}^*$]{rSSTP}\, 
\min\, &\sum_{a\in A} (c_a^0 - c_a^*) z_a^0 + \sum_{k\in\cK} p^k \sum_{a\in A} c_a^k y_a^k \nonumber \\
\text{s.t. } (z^0, y^{1\dots K}) &\text{ satisfies \eqref{rSSTP:dcut2:cuts:first:stage}--\eqref{rSSTP:dcut2:yk:integrality}} \nonumber
\end{align}

If $c_a^0 < c_a^* := \sum_{k\in\cK} p^k c_a^k $ holds for all arcs $a\in A$ we can again relax the integrality restrictions on the first-stage variables without losing overall integrality.
Let $(\text{rSSTP}_\text{dc2}^{\text{rel}:z0})$ 
denote formulation $\model[dc2]{rSSTP}$ with the integrality constraint \eqref{rSSTP:dcut2:z0:integrality} being relaxed to $z^0 \in [0,1]^{|A|}$.

\begin{theorem}
\label{sstp:dcut2:integrality}
If it holds $c_a^0 < c_a^*, \forall a\in A$, the optimum solution to $(\text{rSSTP}_\text{dc2}^{\text{rel}:z0})$  is integer.
\end{theorem}
\begin{proof}
Let $(\tilde z^0, \tilde y^{1\dots K})$ denote an optimum solution to $(\text{rSSTP}_\text{dc2}^{\text{rel}:z0})$  that is non-integer. 
Now consider an arc $\alpha\in A$ with $0< \tilde z_{\alpha}^0 < 1$ defined as follows.
If there exists a fractional arc $(r,j)$ we set $\alpha := (r,j)$.
Otherwise, we set $\alpha:=(i,j)$ such that the directed path $P$ from the root $r$ to vertex $i$ consists only of selected arcs, i.e., $\tilde z_a^0 = 1, \forall a\in P$.
Notice that arc $\alpha$ is well-defined due to constraints \eqref{rSSTP:dcut2:cuts:first:stage}.

We consider three main cases. 
In each case we construct a feasible solution $(\hat z^0, \hat y^{1\dots K})$ with a better objective value than by $(\tilde z^0, \tilde y^{1\dots K})$.
We always start with the solution $(\hat z^0, \hat y^{1\dots K})$ with $\hat z^0 := \tilde z^0, \hat y^{1\dots K} := \tilde y^{1\dots K}$ and describe the necessary modifications.

{\em Case 1:} $\alpha = (i, r)$. 
Since $\alpha$ is an ingoing arc of the root $r$ it is not contained in any directed cut.
Hence, setting $\hat z_\alpha^0 := 0$ and $\hat y_\alpha^k := 0, \forall k\in\cK$, gives a better solution.

{\em Case 2:} $\alpha=(r, j)$.
In this case set $\hat z_\alpha^0 := 1$. 
First, notice that the objective value improves since 
the term in the objective function with respect to arc $\alpha$ is
$
c_\alpha^0 \tilde z_\alpha^0 + \sum_{k\in\cK} p^k c_\alpha^k (\tilde y_\alpha^k - \tilde z_\alpha^0) 
= c_\alpha^0 \tilde z_\alpha^0 + \sum_{k\in\cK} p^k c_\alpha^k (1-\tilde z_\alpha^0) 
=  (c_\alpha^0 - c_\alpha^*) \tilde z_\alpha^0 + c_\alpha^*
$
and $c_\alpha^0 < c_\alpha^*$. 

Second, we argue that the solution $(\hat z^0, \hat y^{1\dots K})$ is feasible.
Since $\hat y^{1\dots K}=\tilde y^{1\dots K}$ we do not need to consider constraints \eqref{rSSTP:dcut2:scenario:cuts}.
Constraints \eqref{rSSTP:dcut2:cuts:first:stage} are only crucial for vertex $j$ since for all other vertices the right-hand side does not change and the left-hand side does not decrease.
For vertex $j$ notice that $z_\alpha^0$ is contained in the left-hand and in the right-hand side of any constraint;
hence, the constraints are still satisfied.  
Constraint \eqref{rSSTP:dcut2:capacity:first:second:stage} is also only interesting for arc $\alpha$;
but since $\tilde z_\alpha^0 > 0$ it holds $\hat y_\alpha^k = \tilde y_\alpha^k = 1$ and the constraint is also still satisfied.

{\em Case 3:} $\alpha=(i, j)$ with $i\not=r, j\not=r$. 
Let $\cL := \{\ell\in V \mid (\ell, j)\in A, \ell \not= i,  \tilde z_{\ell j}^0 > 0 \}$, i.e., $\cL$ is the set of vertices $\ell\not=i$ with a (fractionally) selected arc $(\ell, j)$.

{\em Case 3.1:}  $\cL = \emptyset$.
Hence, arc $\alpha$ is the only ingoing arc of $j$ with $\tilde z_{\cdot,j}^0 > 0$.
In this case we set  $\hat z_\alpha^0 := 1$.

The arguments are similar to Case 2. 
Again, the objective value improves and constraints \eqref{rSSTP:dcut2:scenario:cuts} and \eqref{rSSTP:dcut2:capacity:first:second:stage} are still satisfied. 
Constraints \eqref{rSSTP:dcut2:cuts:first:stage} are again only crucial for vertex $j$ 
and are satisfied due to the properties of arc $\alpha$: 
Recall that we set $\alpha$ such that the directed path $P$ from $r$ to $i$ consists of arcs $a$ with $\tilde z_a^0 = 1, \forall a\in P$. 
Hence, any cut $S$ with $j\in S, r\not\in S$ satisfies $\hat z^0(\delta^-(S)) \ge \tilde z^0(\delta^-(S)) \ge 1 = \hat z^0(\delta^-(j))$.

{\em Case 3.2:}  $\cL \not= \emptyset$. 
Since $\cL \not= \emptyset$ there exists at least one arc $(\ell, j)$ with $\tilde z^0_{\ell j} > 0, \ell \not=i$.

Hence, due to capacity constraints \eqref{rSSTP:dcut2:capacity:first:second:stage} it holds 
$\tilde y^k(\delta^-(j)) = 1 + |\cL| \ge 2$ in any scenario $k\in\cK$.
Since directed cuts have a left-hand side of 1 it is obvious that this solution is non-optimal. 

Now, set $\hat z_{\alpha}^0 := 1$, 
$\hat z_{\ell j}^0 := 0, \forall \ell\in\cL$,
and $\hat y_{\ell j}^k := 0, \forall \ell\in\cL, \forall k\in\cK$.
First, we argue that this solution has a better objective value and afterwards, we discuss its feasibility. 

As discussed in Case 2 increasing $\hat z_{\alpha}^0$ leads to a decrease of the objective value.
Moreover, deleting arcs from the solution by setting $\hat z_{\ell j}^0 := 0, \forall \ell\in\cL$,
and $\hat y_{\ell j}^k := 0, \forall \ell\in\cL, \forall k\in\cK$, improves the objective, too.
Hence, the newly constructed solution has a better objective value.

To show the feasibility of this solution we consider the constraints one by one. 
Capacity constraints \eqref{rSSTP:dcut2:capacity:first:second:stage} are satisfied by construction. 
The directed cuts in the scenarios \eqref{rSSTP:dcut2:scenario:cuts} are satisfied for every valid cut $S\ni j$ since 
$S$ crosses the path $P$ or arc $\alpha$ where each arc $a\in P\cup\alpha$ has a value $\hat y_a^k = 1, \forall k\in\cK$,  such that
it holds $\hat y^k(\delta^-(S))\geq 1$. 
All other valid cuts $S\not\ni j$ are still satisfied since the arc variables crossing the cuts are not modified.  

Last but not least, we have to consider constraints \eqref{rSSTP:dcut2:cuts:first:stage}; here, the arguments are very similar. 
Consider any valid cut $S$ for constraint \eqref{rSSTP:dcut2:cuts:first:stage}. 
If $j\not\in S$ the constraint is still satisfied since the related arc variables are unchanged. 
In case $j\in S$ the cut $S$ crosses $P\cup \alpha$ such that (i) $\hat z^0(\delta^-(S))\geq 1$. 
Since arc costs are non-negative and the right-hand side of the directed cuts is 1 any optimum solution satisfies (ii) $z^0(\delta^-(v)) \leq 1, \forall v\in V$. 
We modified $z^0$ such that (iii) $\hat z^0(\delta^-(v)) = 1, \forall v\in V$. 
Combining (i)--(iii) shows that constraints \eqref{rSSTP:dcut2:cuts:first:stage} are satisfied. 
\qed
\end{proof}

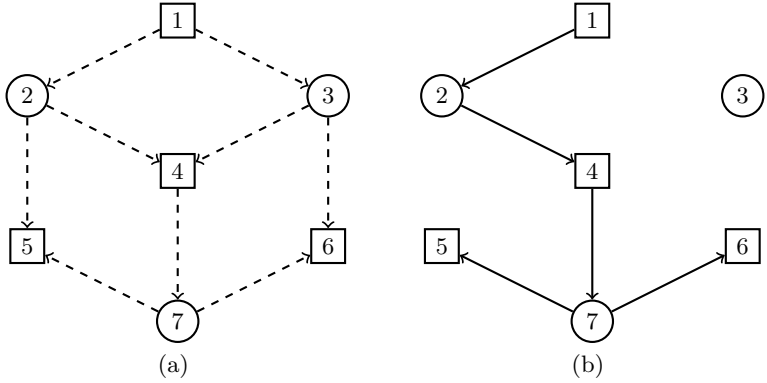
\begin{figure}[tb]
\centering
\subfigure[]{
\begin{tikzpicture}[scale=1.0]

\node[rectangle] (1) at (11,11) {$1$};
\node[circle] (2) at (9,10) {$2$};
\node[circle] (3) at (13,10) {$3$};
\node[rectangle] (4) at (11,9) {$4$};
\node[rectangle] (5) at (9,8) {$5$};
\node[rectangle] (6) at (13,8) {$6$};
\node[circle] (7) at (11,7) {$7$};

\draw[arcr,dashed] (1) -- node[right] {$$} node[left] {} (2);
\draw[arcr,dashed] (1) -- node[left] {$$} node[right] {} (3);

\draw[arcr,dashed] (2) -- node[above] {$$} node[below] {} (4);
\draw[arcr,dashed] (2) -- node[above] {$$} node[below] {} (5);

\draw[arcr,dashed] (3) -- node[above] {$$} node[below] {} (4);
\draw[arcr,dashed] (3) -- node[above] {$$} node[below] {} (6);

\draw[arcr,dashed] (4) -- node[above] {$$} node[below] {} (7);

\draw[arcr,dashed] (7) -- node[above] {$$} node[below] {} (5);
\draw[arcr,dashed] (7) -- node[above] {$$} node[below] {} (6);

\end{tikzpicture}
}
\hspace{0.5cm}
\subfigure[]{
\begin{tikzpicture}[scale=1.0]

\node[rectangle] (1) at (11,11) {$1$};
\node[circle] (2) at (9,10) {$2$};
\node[circle] (3) at (13,10) {$3$};
\node[rectangle] (4) at (11,9) {$4$};
\node[rectangle] (5) at (9,8) {$5$};
\node[rectangle] (6) at (13,8) {$6$};
\node[circle] (7) at (11,7) {$7$};

\draw[arcr] (1) -- node[right] {$$} node[left] {} (2);

\draw[arcr] (2) -- node[above] {$$} node[below] {} (4);

\draw[arcr] (4) -- node[above] {$$} node[below] {} (7);

\draw[arcr] (7) -- node[above] {$$} node[below] {} (5);
\draw[arcr] (7) -- node[above] {$$} node[below] {} (6);

\end{tikzpicture}
}
\caption{
Instance for the STP where the directed cut formulation has an integrality gap of $10/9$.  
All edge costs are 1 and terminals are drawn as rectangles.
(a) shows the optimum fractional solution (dashed arcs are set to 0.5) whereas (b) depicts an optimum integer solution.
This graph can be used to construct an rSSTP-instance where the optimum solution to $(\text{rSSTP}_\text{dc1}^{\text{rel}:z0})$ is fractional but 
$(\text{rSSTP}_\text{dc2}^{\text{rel}:z0})$ is integer, cf.\ text.
}
\label{figure:sstp:rsstp:first:stage:integer}
\end{figure}

We like to shortly revisit the first directed cut formulation $\model[dc1]{rSSTP}$ and show that $(\text{rSSTP}_\text{dc1}^{\text{rel}:z0})$ does not have the latter property; 
let $(\text{rSSTP}_\text{dc1}^{\text{rel}:z0})$ denote formulation $\model[dc1]{rSSTP}$ with relaxed first-stage variables $z^0\in[0,1]^{|A|}$.
An example is given by Figure \ref{figure:sstp:rsstp:first:stage:integer} (a).
The corresponding undirected graph depicts a classical instance for the deterministic STP 
(cf.\ e.g., \cite{PolzinDaneshmand2009ApporachesSTP}) 
where the directed cut formulation has an integrality gap---here it is $10/9$.
Now, consider an rSSTP-instance on that graph that contains one scenario with the four terminals $\{1, 4, 5, 6\}$ and with vertex 1 being the root $r$. 
Moreover, let the cost in the first stage be 1 for each edge and in the scenario 2 for each edge such that $(\text{rSSTP}_\text{dc1}^{\text{rel}:z0})$ connects all terminals already in the first stage.
Figure \ref{figure:sstp:rsstp:first:stage:integer} (a) gives the optimum solution with cost $4.5$ for model $(\text{rSSTP}_\text{dc1}^{\text{rel}:z0})$ where each dashed arc $a$ is set to $z^0_a := 0.5$;
moreover, $z^{1\dots K} := \mathbf{0}$ is integer. 
Figure \ref{figure:sstp:rsstp:first:stage:integer} (b) depicts the first stage of an optimum solution with cost 5 for the described rSSTP-instance which is also the optimum solution to $(\text{rSSTP}_\text{dc2}^{\text{rel}:z0})$.

\paragraph{Directed flow formulations for the rSSTP.}
We close the discussion on formulations for the rooted stochastic Steiner tree problem by introducing a polynomially sized model.
This formulation is again flow-based.
Compared to the previously introduced flow formulations it requires additional node variables $w_v^0 \in\{0,1\}, \forall v\in V$, and additional first-stage flow variables $f_{ij}^{0,v}, \forall v\in V_r, \forall (i,j)\in A$, for ensuring the first-stage tree.

The description of the formulation is split into several parts for better readability.
First, we introduce the variables. 
The solution is represented by arc variables $z^0$ for the first stage and $y^{1\dots K}$ for the $K$ scenarios. 
We again use capacity constraints to ensure that each first-stage arc is also used in each scenario.
Hence, we have the same identifiers $y^{1\dots K}$ for the second stage.

As for the semi-directed flow formulations we have flow variables $f_{ij}^{k,t}$ for each scenario $k\in\cK$, terminal $t\in T_r^k$, and arc $(i,j)\in A$. 
Moreover, we use the already mentioned flow variables $f_{ij}^{0,v}$ and binary node variables $w_v^0$ for the first stage.
\begin{align}
f^0 &\in [0,1]^{|V_r|\cdot|A|} \label{rsstp:dflow:f0:in:01} \\
f^{1\dots K} &\in [0,1]^{|A|\cdot t_r^*}  \\
z^0 &\in \{0, 1\}^{|A|}\\
w^0 &\in\{0,1\}^{|V_r|} \label{rsstp:dflow:w0:integrality}\\
y^{1\dots K} &\in \{0, 1\}^{|A|\cdot K} \label{rsstp:dflow:integrality:yk}
\end{align}

The constraints which contain first-stage variables are given as follows.
Thereby, $w_v^0=1$ implies that vertex $v$ is contained in the first-stage tree.
In this case a flow of unit one needs to be send from the root to this vertex.
This is ensured by the classical flow conservation constraints \eqref{rsstp:dflow:first:stage:low:conservation}; here with the right-hand side $w_i^0$ and $-w_i^0$, respectively.
Constraints \eqref{rsstp:dflow:first:stage:variables}  ensure the correct assignment of the node variables.
\begin{align}
z_{ij}^0 &\ge f_{ij}^{0,v} 
	\hphantom{z^0(\delta^-(v))  } \forall v\in V_r, \forall (i,j)\in A \label{rsstp:dflow:first:stage:capacity:arcs:flow} \\
w_v^0 &\ge z^0(\delta^-(v))  
	\hphantom{f_{ij}^{0,v} } \forall v\in V_r \label{rsstp:dflow:first:stage:variables}\\ 
\sum_{(h,i) \in A} f_{hi}^{0,v} - \sum_{(i,j) \in A} f_{ij}^{0,v} &=
\left.\begin{cases}
w_v^0, & \text{if } i = r\\
-w_v^0, & \text{if } i = v\\
0, & \text{otherwise}
\end{cases} \right\}
\, \forall v\in V_r, \forall i \in V \label{rsstp:dflow:first:stage:low:conservation}
\end{align}

Again, we use capacity constraints \eqref{rsstp:dflow:capacity:constraint:first:second:stage} to ensure that each first-stage arc is also used in each scenario.
These constraints link the first and second stage and they are the only constraints using both first and second-stage variables.
\begin{align}
y_{ij}^k &\ge z_{ij}^0  \qquad \forall k \in\cK, \forall (i,j)\in A \label{rsstp:dflow:capacity:constraint:first:second:stage}
\end{align}

The remaining constraints are identical to the constraints in the semi-directed flow formulation. 
They ensure that all arcs used for routing flow are also purchased in the objective function and that the constructed flow is valid.
\begin{align}
y_{ij}^k &\ge f_{ij}^{k,t}  \qquad \forall k \in\cK, \forall (i,j)\in A, \forall t\in T^k_r \label{rsstp:dflow:capacity:constraint:scenario}\\
\sum_{(h,i) \in A} f_{hi}^{k,t} - \sum_{(i, j) \in A} f_{ij}^{k,t} &=
\left.\begin{cases}
1, & \text{if } i = r\\
-1, & \text{if } i = t\\
0, & \text{otherwise}
\end{cases} \right\}
\begin{array}{l}
\forall k\in\cK, \forall t\in T^k_r,\\  \forall i \in V 
\end{array}
\label{rsstp:dflow:flow:conservation}
\end{align}

Finally, the directed flow-based formulation reads as follows:
\begin{align}
\model[df]{rSSTP}\, 
\min\, \sum_{a\in A}c_{a}^0 z_{a}^0 +& \sum_{k\in\cK} p^k \sum_{a\in A}c_{a}^k (y_{a}^k - z_{a}^0) \nonumber \\
\text{s.t. }
(z^0, y^{1\dots K}, w^0, f) \text{ satisfies }&\text{\eqref{rsstp:dflow:f0:in:01}--\eqref{rsstp:dflow:flow:conservation}}\nonumber
\end{align}

\begin{observation}
Formulation $\model[df]{rSSTP}$ models the rooted stochastic Steiner tree problem correctly.
\end{observation}
%

\subsection{Additional constraints}
\label{section:additional:constraints}

It is possible to expand the formulations for the (r)SSTP by further inequalities which are valid for the deterministic STP as described by, e.g., 
\cite{KochMartinSTP} and \cite{PolzinDaneshmand2001}.

Although the following constraints do not strengthen the models they are all valid for any scenario $k\in\cK$. 
Here, we use variables $y^k$ but the constraints can be used for $\model[sdc1]{SSTP}$ and $\model[dc1]{rSSTP}$ as well. 
\begin{alignat}{2}
y^k_{ij} + y^k_{ji} &\le 1 \quad\quad&&\forall e=\{i,j\}\in E \label{SSTP:sdcut:SEC2}\\
y^k(\delta^-(r^k)) &= 0 \quad\quad&&\label{SSTP:indegree:root:0}\\
y^k(\delta^+(r^k)) &\ge 1 && \label{SSTP:sdcut:outdegree:root} \\
y^k(\delta^-(v)) &= 1 &&\forall v\in T^k_r \label{SSTP:sdcut:indegree:terminal} \\
y^k(\delta^-(v)) &\le 1 &&\forall v\in V_r^k\without T_r^k \label{SSTP:indegree:nonterminal:leq:1}
\end{alignat}

By using straight-forward modifications constraints  
\eqref{SSTP:sdcut:SEC2}, \eqref{SSTP:indegree:root:0}, and \eqref{SSTP:indegree:nonterminal:leq:1} are also valid for the first stage of the rSSTP models. 

\paragraph{Flow-balance constraints.} 
These constraints are deviated from the flow-conservation condition and relate the in- and outdegree of non-terminal vertices. 
E.g., \cite{PolzinDaneshmand2001} showed that constraints \eqref{flow:balance:constraints} strengthen the directed cut- and flow-based formulations of the STP.
\begin{align}
z(\delta^+(v)) \ge z(\delta^-(v)) \hspace{25pt}\forall v \in V\wo T \label{flow:balance:constraints}
\end{align}

However, these constraints are not valid for the stochastic models. 
Since first-stage solutions might contain irrelevant parts w.r.t.\ one particular scenario $k$, i.e., there might be parts of the first-stage solution that can be pruned without violating the feasibility of the solution in scenario $k$, these constraints would enforce the selection of unnecessary arcs. 
Notice that this holds for both the semi-directed and directed formulations (in the first and second stage, too).

\section{Strength of the formulations}
\label{section:sstp:ip:formulations:strength}

This section provides a comparison of the introduced formulations from a polyhedral point of view. 
In the first part we consider the undirected (Section \ref{section:sstp:ip:formulations:strength:undirected}) 
and semi-directed formulations (Section \ref{section:sstp:ip:formulations:strength:semi:directed}) 
for the SSTP and Section \ref{sstp:ip:formulations:strength:directed} focusses on the directed models for the rooted version.

\subsection{Undirected formulations for the SSTP}
\label{section:sstp:ip:formulations:strength:undirected}

We start by comparing the undirected formulations based on cuts and flows, respectively.
The related polytopes of the relaxed formulations are denoted by
\begin{align}
\polytope[uc]{SSTP} &= \left\{x^{0\dots K} \in [0,1]^{|E|\cdot(K+1)} \midresize x^{0\dots K} \text{ satisfies \eqref{SSTP:ucut:undirected:cuts}}  \right\} \nonumber\\ 
\polytope[uf]{SSTP} &= \left\{(x^{0\dots K}, f) \in [0,1]^{|E|\cdot(K+1)}\times [0,1]^{|A|\cdot t_r^*} \right|\nonumber\\
&\hphantom{= \left\{x^{0\dots K} \in [0,1]^{|E|\cdot(K+1)} \right|}\left.\,(x^{0\dots K}, f) \text{ satisfies \eqref{sstp:uflow:capacity:constraint:scenario}, \eqref{sstp:uflow:flow:conservation}}  \right\}. \nonumber
\end{align}

In order to compare the formulations we project the variables of the flow formulation onto the space of undirected edge variables, i.e.,
\[
\projection{x^{0\dots K}}{\polytope[uf]{SSTP}} = \left\{x^{0\dots K} \midresize
\exists f\colon (x^{0\dots K}, f) \in \polytope[uf]{SSTP}\right\}. \nonumber
\]

As for the undirected cut-based and flow-based formulations of the deterministic STP the two formulations for the SSTP are equivalently strong.

\begin{lemm}
$\projection{x^{0\dots K}}{\polytope[uf]{SSTP}} = \polytope[uc]{SSTP}$.
\end{lemm}
\begin{proof}
This lemma follows directly from the classical max flow = min cut theorem, applied to each scenario.
If there is a flow of one unit from the root node to each terminal then every cut separating the terminal from the root node is satisfied.
On the other hand, if every undirected cut is satisfied it is easy to find a feasible flow from the root node to every terminal using exactly those edges.
In both models either first- or second-stage edges can be used. 
\qed \end{proof}

\subsection{Semi-directed formulations for the SSTP}
\label{section:sstp:ip:formulations:strength:semi:directed}

Before comparing the formulations we expand the semi-directed cut formulations by \emph{subtour elimination constraints of size two (SEC2)} in the second stage;
constraints \eqref{sstp:sdcut1:additional:sec2} are added to
$\model[sdc1]{SSTP}$ and 
\eqref{sstp:sdcut2:additional:sec2} to
$\model[sdc2]{SSTP}$, respectively:
\begin{align}
z_{ij}^k + z_{ji}^k &\le 1\hspace{25pt} \forall k\in\cK, \forall (i,j)\in A \label{sstp:sdcut1:additional:sec2} \\
y_{ij}^k + y_{ji}^k &\le 1\hspace{25pt} \forall k\in\cK, \forall (i,j)\in A \label{sstp:sdcut2:additional:sec2}
\end{align}

We introduce the additional constraints to make the comparison of polytopes easier.
Although these constraints cut the polytopes of the LP relaxations they are not binding, i.e., any optimum solution satisfies the SEC2's anyway.

Then, the polytopes of the relaxed cut formulations are denoted by

\begin{align}
\polytope[sdc1]{SSTP} &= \left\{(x^0, z^{1\dots K}) \in [0,1]^{|E|} \times [0,1]^{|A|\cdot K} \right| \nonumber\\
&\hspace{60pt}\left.(x^0, z^{1\dots K}) \text{ satisfies \eqref{SSTP:sdcut1:semi:directed:cuts}, \eqref{sstp:sdcut1:additional:sec2}}  \right\} \nonumber\\
\polytope[sdc2]{SSTP} &= \left\{(x^0, y^{1\dots K}) \in [0,1]^{|E|} \times [0,1]^{|A|\cdot K} \right| \nonumber\\
&\hspace{60pt}\left.(x^0, y^{1\dots K}) \text{ satisfies \eqref{SSTP:sdcut2:directed:cuts}, \eqref{SSTP:sdcut2:capacity:first:second:stage}, \eqref{sstp:sdcut2:additional:sec2}}  \right\} \nonumber
\end{align}

Again, we consider the projections onto the space of undirected edge variables $x^{0\dots K}$:
\begin{align}
\projection{x^{0\dots K}}{\polytope[sdc1]{SSTP}} &= \left\{x^{0\dots K} \right| 
\exists z^{1\dots K}\colon (x^0, z^{1\dots K}) \in \polytope[sdc1]{SSTP}, \nonumber\\
 &\hphantom{= \left\{x^{0\dots K} \right|} \left.\, x_e^k = z_{ij}^k + z_{ji}^k, \forall k\in\cK, \forall e=\{i,j\}\in E \right\} \nonumber\\
\projection{x^{0\dots K}}{\polytope[sdc2]{SSTP}} &= \left\{x^{0\dots K} \right| 
\exists y^{1\dots K}\colon (x^0, y^{1\dots K}) \in \polytope[sdc2]{SSTP}, \nonumber\\
 &\hphantom{= \left\{x^{0\dots K} \right| } \left.\, x_e^k = y_{ij}^k + y_{ji}^k - x_e^0, \forall k\in\cK, \forall e=\{i,j\}\in E \right\} \nonumber
\end{align}

\begin{figure}[tb]
\centering
\begin{tikzpicture}[scale=1.0]

\node[rectangle,minimum width=6cm, minimum height=1cm] (r0) at (11.5,10) {};
\node[rectangle,minimum width=2.5cm] (0) at (10,10) {$\model[uc]{SSTP}$};
\node[rectangle,minimum width=2.5cm] (1) at (11.5,12) {$\model[sdc1]{SSTP}$};
\node[rectangle,minimum width=2.5cm] (2) at (8,14) {$\model[sdc2]{SSTP}$};
\node[rectangle,minimum width=2.5cm] (3) at (11.5,14) {$(\text{SSTP}_{\text{sdc2}^*})$};

\node[rectangle,minimum width=10cm, minimum height=1cm] (r1) at (11.5,14) {};

\node[rectangle,minimum width=2.5cm] (4) at (13,10) {$\model[uf]{SSTP}$};
\node[rectangle,minimum width=2.5cm] (5) at (15,14) {$\model[sdf]{SSTP}$};

\draw[arc] (r0) -- (1);
\draw[arc] (1) -- (r1);

\draw[dashedEdge] (2) -- (3);

\draw[dashedEdge] (0) -- (4);
\draw[dashedEdge] (3) -- (5);

\end{tikzpicture}
\caption{
Hierarchy of undirected and semi-directed formulations for the SSTP.
The dashed line and the additional clusters specify that formulations are equivalent.
An arrow indicates that the target cluster contains stronger formulations than the formulations in the source cluster.
}
\label{figure:sstp:semi:directed:formulations:strength}
\end{figure}
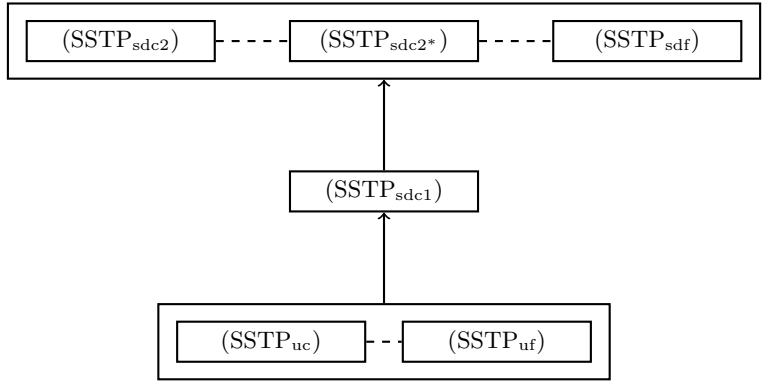

We start by comparing the undirected and the first semi-directed cut formulation.
Not surprising, the additional directed parts of the formulation make it stronger.  

\begin{theorem}
\label{lemma:sstp:u:vs:sd1}
$\polytope[uc]{SSTP} 
\supsetneq 
\projection{x^{0\dots K}}{\polytope[sdc1]{SSTP}}$, i.e., 
the semi-directed cut-based formulation $\model[sdc1]{SSTP}$ is stronger than the undirected 
formulation $\model[uc]{SSTP}$.
\end{theorem}
\begin{proof}
Let $(\tilde x^0, \tilde z^{1\dots K}) \in \polytope[sdc1]{SSTP}$ and 
set $\hat x^0 := \tilde x^0, \hat x_e^k := \tilde z_{ij}^k + \tilde z_{ji}^k, \forall k\in\cK, \forall e=\{i,j\} \in E$. 
We obtain a solution $\hat x^{0\dots K}$ for $\model[uc]{SSTP}$; its validity is discussed in the following.

Bounds of the first-stage variables $\hat x^0$ are obviously satisfied.
Moreover, it clearly holds $\hat x_e^k \ge 0$ and due to constraints \eqref{sstp:sdcut1:additional:sec2}: $\hat x_e^k = \tilde z_{ij}^k + \tilde z_{ji}^k \le 1$.
Hence,  $\hat x^{0\dots K} \in [0, 1]^{|E|\cdot (K+1)}$.

We now show that the undirected cuts \eqref{SSTP:ucut:undirected:cuts} are also satisfied by $\hat x^{0\dots K}$.
Let $S\subseteq V$ represent a feasible cut set in scenario $k\in\cK$, i.e., $\emptyset \not=S\cap T^k\not= T^k$. 
Since cuts in $\model[sdc1]{SSTP}$ are semi-directed and ingoing we assume w.l.o.g.\ that it holds $r^k\not\in S$.
Otherwise one can simply consider the complementary set $V\wo S$, 
since $\delta(S) = \delta(V\wo S)$ and then, 
it holds $r^k\not\in (V\wo S)$.
\begin{align}
(\hat x^0 + \hat x^k)(\delta(S)) &=\sum_{e\in\delta(S)} \hat x_e^0 + \hat x_e^k \nonumber\\
&=\sum_{e\in\delta(S)} \tilde x_e^0 + \sum_{\{i,j\}\in\delta(S)} \tilde z_{ij}^k + \tilde z_{ji}^k \nonumber\\
&\ge\tilde x^0(\delta(S)) + \tilde z^k(\delta^-(S)) \ge 1\nonumber
\end{align}

The last inequality holds since $(\tilde x^0, \tilde z^{1\dots K})$ satisfies constraint \eqref{SSTP:sdcut1:semi:directed:cuts} for cut set $S$. 

Intuitively, the strict inequality of the formulations results from the directed arcs in the scenarios and the strength of the directed cut formulation for the deterministic STP.
Figure \ref{figure:sstp:u:vs:sd1} gives a small example with this property where 
everything is purchased in the second stage and the relaxed semi-directed model gives a better lower bound.
\qed \end{proof}

\begin{figure}[tb]
\centering
\begin{tikzpicture}[scale=1.0]

\node[rectangle] (1) at (11,12) {$1$};
\node[rectangle] (2) at (9,10) {$2$};
\node[rectangle] (3) at (13,10) {$3$};

\draw[basicEdge] (1) -- node[right] {$e_1$} node[left] {10/1} (2);
\draw[basicEdge] (1) -- node[left] {$e_2$} node[right] {10/1} (3);
\draw[basicEdge] (2) -- node[above] {$e_3$} node[below] {10/1} (3);

\end{tikzpicture}
\caption{
Example where the LP relaxation of $(\text{SSTP}_{\text{sdc1}})$ gives a better lower bound than
$(\text{SSTP}_{\text{uc}})$.
There is one scenario and all vertices are terminals. 
Edge costs for the first stage are all 10 and for the scenario 1.  
Both formulations purchase edges only in the second stage. 
The optimum solution to the undirected formulation has cost $1.5$ with $x_e^1 = 0.5, \forall e\in E$. 
Since there is no valid orientation using $0.5$ of each edge the semi-directed formulation selects two arcs in the second stage to connect the two remaining vertices to a root node leading to overall cost 2, e.g., for root node 1 set $z_{(1,2)}^1 = z_{(1,3)}^1 = 1$.
}
\label{figure:sstp:u:vs:sd1}
\end{figure}
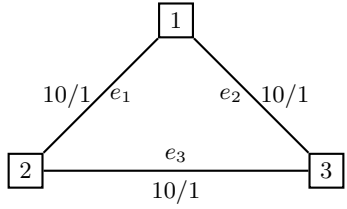

The following theorem shows that formulation $\model[sdc2]{SSTP}$ is stronger than formulation $\model[sdc1]{SSTP}$.

\begin{theorem}
\label{lemma:sstp:sd1:vs:sd2}
$\projection{x^{0\dots K}}{\polytope[sdc1]{SSTP}} 
\supsetneq 
\projection{x^{0\dots K}}{\polytope[sdc2]{SSTP}}$.
\end{theorem}
\begin{proof}
Let $(\tilde x^0, \tilde y^{1\dots K}) \in \polytope[sdc2]{SSTP}$ 
and set
$\hat x_e^0 := \tilde x_e^0, \forall e\in E$,   
and $\hat x_e^k := \tilde y_{ij}^k + \tilde y_{ji}^k - \tilde x_e^0, \forall k\in \cK, \forall e=\{i,j\}\in E$.
We argue  that $\hat x^{0\dots K} \in \projectionS{x^{0\dots K}}{\polytope[sdc1]{SSTP}}$ 
by showing that there exists a variable assignment $\hat z^{1\dots K} \in [0,1]^{K\cdot |A|}$ such that 
$(\hat x^0, \hat z^{1\dots K}) \in \polytope[sdc1]{SSTP}$.

This solution is obtained by transforming $(\tilde x^0, \tilde y^{1\dots K})$ into a feasible $\model[sdc1]{SSTP}$-solution.
Thereby, the parameter $\alpha_{ij}^k \in [0,1], \forall k \in\cK, \forall (i,j)\in A$, is used:
\[
\alpha_{ij}^k := 
\begin{cases}
\frac{\tilde y_{ij}^k}{\tilde y_{ij}^k + \tilde y_{ji}^k} & \text{if}\ \tilde y_{ij}^k + \tilde y_{ji}^k > 0 \\
0 & \text{otherwise}.
\end{cases}
\]

This parameter allows us to split up the first-stage values among the two corresponding directed arcs, independent for each scenario.
With $\alpha$ at hand the directed arc variables are set to $\hat z_{ij}^k := \tilde y_{ij}^k - \alpha_{ij}^k \tilde x_{e}^0, \forall k \in\cK, \forall (i,j)\in A$, with $e =\{i,j\} \in E$.
 
First we show that this is a valid projection.
Notice that $\forall e=\{i,j\}\in E, \forall k\in\cK$: $\alpha_{ij}^k + \alpha_{ji}^k \in\{0,1\}$; if $\tilde y_{ij}^k + \tilde y_{ji}^k > 0$ this value is 1 and 0 otherwise.
Now, consider edge $e = \{i,j\}\in E$ in scenario $k\in\cK$ with $\tilde y_{ij}^k + \tilde y_{ji}^k > 0$. 
Then, $\hat z_{ij}^k + \hat z_{ji}^k =  \tilde y_{ij}^k - \alpha_{ij}^k \tilde x_e^0 + \tilde y_{ji}^k - \alpha_{ji}^k \tilde x_e^0 = \tilde y_{ij}^k + \tilde y_{ji}^k - \tilde x_e^0$.
In case $\alpha_{ij}^k = \alpha_{ji}^k = 0$, due to $\tilde y_{ij}^k + \tilde y_{ji}^k = 0$ and constraints \eqref{SSTP:sdcut2:capacity:first:second:stage}, i.e., 
$y_{ij}^k + y_{ji}^k \ge x_e^0$,  
it follows $\tilde x_e^0 = 0$.
Hence, it always holds $\hat z_{ij}^k + \hat z_{ji}^k = \tilde y_{ij}^k + \tilde y_{ji}^k - \tilde x_e^0, \forall k\in\cK, \forall e=\{i,j\}\in E$.

Now we are able to prove $\hat x^{0\dots K} \in \projectionS{x^{0\dots K}}{\polytope[sdc1]{SSTP}}$.
Due to the preceding discussion it is clear that the subtour elimination constraints \eqref{sstp:sdcut1:additional:sec2} are satisfied.
Moreover, it obviously holds $\hat x_e^0 \in [0,1], \forall e\in E$. 

Next, we consider the bounds for 
the directed arc variables $\hat z_{ij}^k, \forall k\in\cK, \forall (i,j)\in A$. 
$\hat z_{ij}^k \le 1$ holds since $\hat z_{ij}^k \le \tilde y_{ij}^k \le 1$. 
Non-negativity can be seen by considering two cases. 
(i) If $\alpha_{ij}^k > 0$:
\[
\hat z_{ij}^k 
= \tilde y_{ij}^k - \alpha_{ij}^k \tilde x_e^0
= \tilde y^k_{ij} - \tilde x_e^0 \frac{\tilde y^k_{ij}}{\tilde y^k_{ij} + \tilde y^k_{ji}} 
= \tilde y^k_{ij} \left( 1- \overbrace{\frac{\tilde x_e^0}{\tilde y^k_{ij} + \tilde y^k_{ji}} }^{\le1}\right) 
\ge 0.
\]

Inequality $\frac{\tilde x_e^0}{\tilde y^k_{ij} + \tilde y^k_{ji}} \le1$ is true due to capacity constraints \eqref{SSTP:sdcut2:capacity:first:second:stage}.
(ii) If $\alpha_{ij}^k = 0$ the non-negativity follows directly since $\hat z^k_{ij} = \tilde y^k_{ij} \ge 0$. 

It remains to show that
a valid cut $S\subseteq V_r$ in scenario $k\in\cK$ is satisfied by $(\hat x^0, \hat z^{1\dots K})$:
\begin{align}
(\hat x^0 + \hat  z^k)(\delta^-(S)) = \sum_{(i,j)\in \delta^-(S)} \hat x_{\{i,j\}}^0 + \hat z_{ij}^k &= \sum_{(i,j)\in \delta^-(S)} \tilde x_{\{i,j\}}^0 + \tilde y_{ij}^k - \alpha_{ij}^k \tilde x_{\{i,j\}}^0 \nonumber\\
&= \sum_{(i,j)\in \delta^-(S)} (1-\alpha_{ij}^k) \tilde x_{\{i,j\}}^0 + \tilde y_{ij}^k \nonumber\\
&\ge \sum_{(i,j)\in \delta^-(S)} \tilde y_{ij}^k \ge 1 \nonumber
\end{align}

The last inequality is true due to the validity of solution $\tilde y^k$ for scenario $k$ and constraints \eqref{SSTP:sdcut2:directed:cuts}.
This completes the ``$\supseteq$''-part of the proof.

An example showing the strict inequality can be constructed by exploiting the different \emph{meaning} of the first-stage variables.
In formulation $\model[sdc1]{SSTP}$ a first-stage edge $e=\{i,j\}$ contributes its value to cuts in both directions, i.e., $\delta^-(S)$ and $\delta^+(S)$.
Contrarily, a feasible solution for formulation $\model[sdc2]{SSTP}$ has to find an orientation for this edge and distribute its value to the related arcs.
In a sloppy way, the same edge has a lesser value in the second semi-directed formulation.

Hence, the same example from Figure \ref{figure:sstp:u:vs:sd1} can be utilized to show the strict inequality; one simply has to set edge costs to 1 for all first-stage and 10 for  the scenario edges, respectively.
There is still one scenario with all three vertices being terminals.
Then, formulation $\model[sdc1]{SSTP}$ selects all three edges at $0.5$ in the first stage satisfying all cuts in the scenario.
On the other hand, this solution is not valid for $\model[sdc2]{SSTP}$ and there is none with overall cost $1.5$.
\qed \end{proof}

To complete the hierarchy of SSTP formulations given in Figure \ref{figure:sstp:semi:directed:formulations:strength} it remains to show the equivalence of the semi-directed flow and cut-based formulations.
To give the formal proof we denote the polytope of the relaxed flow formulation and the projection onto the same variable space as follows.
\begin{align}
\polytope[sdf]{SSTP} &= \left\{(x^0, y^{1\dots K}, f) \in [0,1]^{|E|} \times [0,1]^{|A|\cdot K} \times [0,1]^{|A|\cdot t_r^*} \right| \nonumber\\
&\hphantom{= \left\{x^{0\dots K} \right|\ } \left.\, (x^0, y^{1\dots K}, f) \text{ satisfies \eqref{sstp:uflow:flow:conservation}, \eqref{sstp:sdflow:capacity}, \eqref{sstp:sdflow:capacity:first:second:stage}}  \right\} \nonumber\\
\projection{(x^0, y^{1\dots K})}{\polytope[sdf]{SSTP}} &= \left\{(x^0, y^{1\dots K})\ \middle|\,
\exists f\colon (x^0, y^{1\dots K}, f) \in \polytope[sdf]{SSTP}\right\} \nonumber
\end{align}

The stronger semi-directed cut and flow formulations are equivalent. 
This result is mainly a consequence of the relationship of the deterministic STP formulations.

\begin{lemm}
$\polytope[sdc2]{SSTP} 
= 
\projection{(x^0, y^{1\dots K})}{\polytope[sdf]{SSTP}}$.
\end{lemm}
\begin{proof}
Restricting the models to one particular scenario, i.e., for one $k\in\cK$: $y^k$ or $(y^k, f^k)$, respectively, results in the related   
cut- and flow-based formulations for the deterministic  STP.
Since the formulations for the deterministic STP are equivalent and the remaining parts of the stochastic models are identical the lemma follows.
\qed
\end{proof}

\subsection{Directed formulations for the rSSTP}
\label{sstp:ip:formulations:strength:directed}

To make the comparison of the polytopes easier we add the following constraints to the directed cut formulations:
$\model[dc1]{rSSTP}$ is expanded by both constraints and 
$\model[dc2]{rSSTP}$ only by the second type of constraints \eqref{rsstp:dcuts:z0S:leq1}: 
\begin{align}
z_{ij}^0 + z_{ij}^k&\le 1\hspace{25pt} \forall k\in\cK, \forall (i,j)\in A \label{rsstp:dcut1:sec2:ij:both:stages} \\
z^0(\delta^-(v)) &\le 1\hspace{25pt} \forall v\in V_r \label{rsstp:dcuts:z0S:leq1}
\end{align}

As for the added SEC2's in the semi-directed formulations 
\eqref{rsstp:dcut1:sec2:ij:both:stages} is obviously redundant, too, 
since the right-hand-side of the directed cuts is 1.
The same holds for \eqref{rsstp:dcuts:z0S:leq1}. 

The polytopes of the relaxed formulations are denoted as follows.
\begin{align}
\polytope[dc1]{rSSTP} &= \left\{z^{0\dots K} \in [0,1]^{|A|\cdot(K+1)} \midresize z^{0\dots K} \text{ satisfies \eqref{rSSTP:dcut1:cuts:first:stage}, 
\eqref{rSSTP:dcut1:scenario:cuts}, 
\eqref{rsstp:dcut1:sec2:ij:both:stages}, 
\eqref{rsstp:dcuts:z0S:leq1}}  \right\} \nonumber \\
\polytope[dc2]{rSSTP} &= \left\{(z^0, y^{1\dots K}) \in [0,1]^{|A|\cdot(K+1)} \midresize (z^0, y^{1\dots K}) \text{ satisfies \eqref{rSSTP:dcut2:cuts:first:stage}--\eqref{rSSTP:dcut2:capacity:first:second:stage}, 
\eqref{rsstp:dcuts:z0S:leq1}} \right\} \nonumber
\end{align}

We use a projection for the second formulation to compare both models:
\begin{align}
\projection{z^{0\ldots K}}{\polytope[dc2]{rSSTP}} = \left\{(z^0, z^{1\dots K}) \right|
&(z^0, y^{1\dots K}) \in \polytope[dc2]{rSSTP}, \nonumber\\
&\left.z_{ij}^k = y_{ij}^k - z_{ij}^0, \forall k\in\cK, \forall (i,j)\in A \right\} \nonumber
\end{align}

Both directed cut-based formulations are equivalent:

\begin{theorem}
$\polytope[dc1]{rSSTP} = 
\projectionS{z^{0\ldots K}}{\polytope[dc2]{rSSTP}}$.
\end{theorem}
\begin{proof}
```$\subseteq$'':
Let $\tilde z^{0\dots K} \in\polytope[dc1]{rSSTP}$.
We show that $(\hat z^0, \hat y^{1\dots K}) \in\polytope[dc2]{rSSTP}$ with
$\hat z^0 := \tilde z^0, \hat y^k := \tilde z^k + \tilde z^0, \forall k\in\cK$.

First, we consider the variable bounds.
Since $\hat z^0 = \tilde z^0$ we have $\hat z^0 \in[0,1]^{|A|}$. 
Moreover, $\hat y^k$ is obviously non-negative and due to constraints \eqref{rsstp:dcut1:sec2:ij:both:stages} at most 1: 
$\hat y_{ij}^k = \tilde z_{ij}^k + \tilde z_{ij}^0 \le 1, \forall (i,j)\in A, \forall k\in\cK$.

Second, the directed cuts in the first stage, i.e., constraints \eqref{rSSTP:dcut2:cuts:first:stage}, and constraints \eqref{rsstp:dcuts:z0S:leq1}, are identical in both formulations and hence, they are satisfied.
This is also true for the capacity constraints \eqref{rSSTP:dcut2:capacity:first:second:stage} since
$\hat y_{ij}^k = \tilde z_{ij}^k + \tilde z_{ij}^0 \ge \hat z_{ij}^0, \forall (i,j)\in A, \forall k\in\cK$.

Third, consider a valid cut set $S\subseteq V_r$ in scenario $k\in\cK$.
Since $\tilde z^{0\dots K}$ is a valid solution for $\model[dc1]{rSSTP}$ it satisfies the directed cuts \eqref{rSSTP:dcut1:scenario:cuts}
and leads to $\hat y^k(\delta^-(S)) = (\tilde z^k + \tilde z^0)(\delta^-(S)) \ge 1$.
Hence, the directed cuts \eqref{rSSTP:dcut2:scenario:cuts} are satisfied by $\hat y^{1\dots K}$.

``$\supseteq$'':
The opposite direction is similar.
Let $(\tilde z^0, \tilde y^{1\dots K})\in \projectionS{z^{0\ldots K}}{\polytope[dc2]{rSSTP}}$.
We set 
$\hat z^0 := \tilde z^0, \hat z^k := \tilde y^k - \tilde z^0, \forall k\in\cK$, such that 
$\hat z^{0\dots K}\in \polytope[dc1]{rSSTP}$.

Again, directed cuts in the first stage are obviously satisfied and the variable bounds trivially hold for the first-stage variables.
For the second-stage variables we have $\hat z^k \le \tilde y^k \le \boldsymbol{1}$ and $\hat z^k = \tilde y^k - \tilde z^0 \ge \boldsymbol{0}, \forall k\in\cK$, due to constraints \eqref{rSSTP:dcut2:capacity:first:second:stage}.

The added constraints \eqref{rsstp:dcut1:sec2:ij:both:stages} are satisfied 
since 
$\hat z_{ij}^0 + \hat z_{ij}^k = \tilde z_{ij}^0 + \tilde y_{ij}^k - \tilde z_{ij}^0 = \tilde y_{ij}^k \le 1, \forall (i,j)\in A, \forall k\in\cK$,
and last but not least, a valid cut set $S\subseteq V_r$ in scenario $k\in\cK$ is satisfied 
since 
$(\hat z^0 + \hat z^k)(\delta^-(S)) = (\tilde z^0 + \tilde y^k - \tilde z^0)(\delta^-(S)) = \tilde y^k(\delta^-(S)) \ge 1$. 
\qed
\end{proof}

We close the discussion by comparing the directed flow formulation $\model[df]{rSSTP}$ 
to the second directed cut formulation $\model[dc2]{rSSTP}$. 
\begin{align}
\polytope[df]{rSSTP} = &\left\{(z^0, y^{1\dots K}, w^0, f) \in [0,1]^{|A|\cdot(K+1)} \times [0,1]^{|V_r|} \times [0,1]^{|A|(|V_r| + t_r^*)} \right| \nonumber\\
&\left. \hphantom{(z^0, y^{1\dots K}, w^0, f) \in [0,1]^{|A|\cdot(K+1)}}
(z^0, y^{1\dots K}, w^0, f) \text{ satisfies \eqref{rsstp:dflow:first:stage:capacity:arcs:flow}--\eqref{rsstp:dflow:flow:conservation}} \right\} \nonumber
\end{align}
\begin{align}
\projection{(z^0, y^{1\dots K})}{\polytope[df]{rSSTP}} &= \left\{(z^0, y^{1\dots K})\ \middle|\,
\exists (w^0, f)\colon (z^0, y^{1\dots K}, w^0, f) \in \polytope[df]{rSSTP} \right\} \nonumber
\end{align}

\begin{theorem}
$\polytope[dc2]{rSSTP} = \projection{(z^0, y^{1\dots K})}{\polytope[df]{rSSTP}}$
\end{theorem}
\begin{proof}
``$\subseteq$'':
Let $(\tilde z^0, \tilde y^{1\ldots K})\in \polytope[dc2]{rSSTP}$. 
We use $(\hat z^0, \hat y^{1\ldots K}) := (\tilde z^0, \tilde y^{1\ldots K})$ to construct a solution $(\hat z^0, \hat y^{1\dots K}, \hat w^0, \hat f) \in \polytope[df]{rSSTP}$. 
First, constraints \eqref{rsstp:dflow:capacity:constraint:first:second:stage} are contained in both models and hence satisfied for $(\hat z^0, \hat y^{1\ldots K})$. 
Second, since $w^0$ gives the connected vertices in the first stage we set $\hat w_v^0 := \tilde z(\delta^-(v)), \forall v\in V_r$; 
due to \eqref{rsstp:dcuts:z0S:leq1} we have $\tilde z(\delta^-(v))\in[0,1]$. 
Hence, bounds on $\hat w^0$ are satisfied and moreover, \eqref{rsstp:dflow:first:stage:variables} is satisfied with equality. 
Third, the remaining part of formulation $\model[df]{rSSTP}$ contains the construction of flow: 
we set the flow variables $\hat f$ such that in the first stage a flow with value $\hat w_v^0$ is send from the root to every vertex and in every scenario a flow of value 1 from the root to every terminal. 
The feasibility and correctness follows again from ``max flow = min cut''. 

``$\supseteq$'':
Let $(\tilde z^0, \tilde y^{1\dots K}, \tilde w^0, \tilde f) \in \polytope[df]{rSSTP}$. 
Again, set $(\hat z^0, \hat y^{1\ldots K}) := (\tilde z^0, \tilde y^{1\ldots K})$.  
First, \eqref{rsstp:dcuts:z0S:leq1} is satisfied for all vertices $v\in V_r$ since $\hat z^0(\delta^-(v)) \le \tilde w_v^0$ due to \eqref{rsstp:dflow:first:stage:variables}. 
Second, due to \eqref{rsstp:dflow:first:stage:capacity:arcs:flow}--\eqref{rsstp:dflow:first:stage:low:conservation} there is a flow with value $\tilde w^0_v$ in the first stage from the root to a vertex $v\in V_r$ with $\tilde w^0_v > 0$ and moreover, arcs used for routing flow are selected by $\tilde z^0$ through \eqref{rsstp:dflow:first:stage:capacity:arcs:flow}. 
Hence, again due to ``max flow = min cut'', the directed cuts \eqref{rSSTP:dcut2:cuts:first:stage} are satisfied for $\hat z^0$ for all $v\in V_r$. 
The same holds for the directed cuts in the scenarios \eqref{rSSTP:dcut2:scenario:cuts} and variables $\hat y^k, \forall k\in\cK$. 
Last but not least, \eqref{rSSTP:dcut2:capacity:first:second:stage} is satisfied since the constraints are contained in both models. 
\qed
\end{proof}

\bibliographystyle{plain}
\bibliography{../StochasticNetworkDesign}

\end{document}